\documentclass[submission,copyright,creativecommons]{eptcs}
\providecommand{\event}{ICLP 2023} % Name of the event you are submitting to

\usepackage[T1]{fontenc}
\usepackage{amsmath,amssymb,amsfonts}
\usepackage{stmaryrd}
\usepackage{mathtools}
\usepackage{tikz} 
\usepackage{listings}
\usepackage{algorithm}
\usepackage[noend]{algpseudocode}
\usepackage{amsthm}
\newcommand{\pred}{\mathit{def}}

\newcommand{\vars}{\mathit{vars}}
\newcommand{\args}{\mathit{args}}
\newcommand{\dom}{\mathit{dom}}

\newcommand\doubleplus{+\kern-1ex+\kern0.8ex}

\newcommand\inargs{\mathit{in}}
\newcommand\outargs{\mathit{out}}

\newcommand\iset{\ensuremath \mathit{ISet}}
\newcommand\pset{\ensuremath \mathbb{P}}
\newcommand\OS{\ensuremath \mathit{OS}}
\newcommand\AP{\ensuremath \mathit{AP}}
\newcommand\oprof{\ensuremath \mathit{opr}}
\newcommand\lfp{\ensuremath \mathit{lfp}}

\newcommand\leafs{\ensuremath \mathit{leafs}}

\usetikzlibrary{calc,decorations.pathmorphing,shapes,fit,trees, automata, positioning, arrows}
 
\newcounter{sarrow}
\newcommand\xrsquigarrow[1]{%
	\stepcounter{sarrow}%
	\mathrel{\begin{tikzpicture}[baseline= {( $ (current bounding box.south) + (0,-0.5ex) $ )}]
			\node[inner sep=.5ex] (\thesarrow) {$\scriptstyle #1$};
			\path[draw,<-,decorate,
			decoration={zigzag,amplitude=1pt,segment length=3mm,pre=lineto,pre length=4pt}] 
			(\thesarrow.south east) -- (\thesarrow.south west);
	\end{tikzpicture}}%
}

\allowdisplaybreaks
\newtheorem{definition}{Definition}
\newtheorem{example}{Example}

\newtheorem{proposition}{Proposition}

\newcommand{\titlerunning}{A Dataflow Analysis for Comparing and Reordering Predicate Arguments}
\newcommand{\authorrunning}{G. Yernaux \& W. Vanhoof}

\hypersetup{
	bookmarksnumbered,
	pdftitle    = {\titlerunning},
	pdfauthor   = {\authorrunning},
	pdfsubject  = {EPTCS},               % Consider adding a more appropriate subject or description
}

\author{Gonzague Yernaux
	\institute{Namur Digital Institute}
	\institute{Faculty of Computer Science\\University of Namur, Belgium}
	\email{gonague.yernaux@unamur.be}
	\and
	Wim Vanhoof 
		\institute{Namur Digital Institute}
	\institute{Faculty of Computer Science\\University of Namur, Belgium}
	\email{wim.vanhoof@unamur.be}
}

\title{A Dataflow Analysis for \\Comparing and Reordering Predicate Arguments}

\begin{document}
\maketitle

\begin{abstract}
	In this work, which is done in the context of a (moded) logic programming language, we devise a data-flow analysis dedicated to computing what we call argument profiles. Such a profile essentially describes, for each argument of a predicate, its functionality, i.e. the operations in which the argument can be involved during an evaluation of the predicate, as well as how the argument contributes to the consumption and/or construction of data values.
	While the computed argument profiles can be useful for applications in the context of program understanding (as each profile essentially provides a way to better understand the role of the argument), they more importantly provide a way to discern between arguments in a manner that is more fine-grained than what can be done with other abstract characterizations such as types and modes. This is important for applications where one needs to identify correspondences between the arguments of two or more different predicates that need to be compared, such as during clone detection. Moreover, since a total order can be defined on the abstract domain of profiles, our analysis can be used for rearranging predicate arguments and order them according to their functionality, constituting as such an essential ingredient for predicate normalization techniques.
	%Logic programming predicates express relations between arguments in a way that sometimes makes the identification of each argument's role non-trivial and perhaps even opaque. This is problematic e.g. in situations where one needs to identify correspondences between the arguments of two or more different predicates. In this work, we devise a dataflow analysis dedicated to computing so-called argument profiles. Such a profile essentially describes the operations in which the argument can be involved during a program execution, as well as the construction of which output arguments it contributes to. The analysis takes as input a moded logic program, where arguments are either input or output and the traditional unification operator is declined in four more specific variants (namely construction, deconstruction, assignment and test). It is able to reorder predicate arguments according to their respective profiles. The analysis output allows to better understand the role of each argument and constitutes one promising way of normalizing the head of program clauses. 
\end{abstract}

%\begin{keywords}
%	Dataflow analysis, Logic Programming, Argument profiles, Ordering Predicate Arguments, Code Normalization
%\end{keywords}

%%%%%%%%%%%%%%%%%%%%%%%%%%%%%%%%%%%%%%%%%%%%%%%%%%%
%%%%%%%%%%%%%%%%%%%%%%%%%%%%%%%%%%%%%%%% 1
%%%%%%%%%%%%%%%%%%%%%%%%%%%%%%%%%%%%%%%%%%%%%%%%%%%
\section{Introduction}\label{sec:intro}

When writing code, subroutines (be it methods, procedures, functions or predicates) and their arguments play an important role, as they constitute the main mechanism by which the programmer can make his or her code modular and general and thus applicable in different contexts. While this is true in any language, it is even more so in declarative languages where modularity is often more fine-grained, resulting in lots of small functions and predicates, and where the lack of iterative control structures makes induction-based control (which itself heavily relies on argument manipulation) the rule rather than the exception~\cite{lp-semantics}. In this work we consider logic programming and thus predicates as the program's main building blocks.

Understanding the source code of a predicate requires thus understanding the role of the arguments involved, and the data flow relations expressed within the code. If one pursues debugging purposes for instance, statically inferring upon which potential instructions (or, in a logic programming context, atoms) each argument does or does not have influence is crucial to better understand the program at hand~\cite{traceclp,slicing}. %This is especially useful in logic programs which, generally speaking, are not renowned for their readability but rather for their capability at being ideal candidates for static analysis regarding data flow thanks to their remarkably simple syntax~\cite{efficient-dataflow}.   
While dataflow analysis is a well-known and indispensable ingredient in applications such as code comprehension~\cite{inputtracer}, compiler optimization~\cite{iterative-dataflow} and automatic parallelization~\cite{parallel-lp}, its potential has, to the best of our knowledge, been less explored in applications such as code normalization, anti-unification and clone detection~\cite{blanker,rattansurvey} which is the prime motivation for the current work. 

Indeed, when comparing predicate definitions during clone detection or anti-unification, one wants to detect as many (dis)similarities as possible~\cite{gen}. It is then often important to consider the right matching between the respective arguments, as the following somewhat contrived example shows. Consider the traditional definition of the \texttt{append/3} predicate and another predicate, \texttt{concat/3}: 
\begin{lstlisting}
append([],L,L).
append([X|Xs],Y,[X|Zs]):- append(Xs,Y,Zs).
\end{lstlisting}
\vspace{-0.1cm}
\begin{lstlisting}
concat(L,[],L).
concat([E|Zs],[E|Es],Y):- concat(Zs,Es,Y).
\end{lstlisting}

Intuitively it is clear that the two predicates define essentially the same ternary relation, where one argument is the concatenation of the two others. The code of the two predicates differs not only in the names of the variables used, but also in the role played by the arguments. Indeed, for an atom \texttt{append($t_1$,$t_2$,$t_3$)} to succeed, $t_3$ must be the concatenation of $t_1$ and $t_2$ whereas for \texttt{concat($t_1$,$t_2$,$t_3$)} to succeed, it is $t_1$ that must be the concatenation of $t_2$ and $t_3$.
For an analysis to detect that one of these predicates is a "clone" -- a textual variant (renaming) of the other \textit{modulo a permutation of the arguments}, it needs to consider potentially all possible argument permutations which adds a non-negligible factor to the complexity of the detection process. In fact, the search for a so-called argument mapping (designating the pairing of corresponding arguments in two predicates) that maximizes the outlined similarity of the involved definitions is one of the key factors rendering a search-based clone detection procedure or, more broadly, the computation of so-called \textit{predicative anti-unification} intractable~\cite{iwsc}. This is especially true when the predicates to be compared are composed of more than a few clauses, since for each suitable argument mapping, there might exist a large number of potential clause mappings that should be explored to find a functional link between the predicates to be compared.

It is not hard to see that the problem of finding a suitable argument mapping can be alleviated by taking adequate abstractions into account. Type- and mode information, for instance, can substantially reduce the number of argument mappings to consider, at least if a sufficient number of arguments are of different type and/or mode. In the example above type information does not really help (as all arguments are supposed to be of the same list type), but using mode information allows to limit the search for corresponding arguments to the subset of input, respectively output arguments of each predicate.

In a more general setting, the question is related to the problem of reordering the arguments in a standard (and preferably unique) way such that arguments playing a similar role (in different predicates) are positioned in similar positions. Ordering arguments is an important aspect of \textit{code normalization}, a process that, generally speaking, aims at restructuring and simplifying code fragments or programs into some kind of \textit{normal} or \textit{canonical} form~\cite{normal-form-asp,normalization} Again, while type and mode information can be used to classify arguments, it is generally not sufficient to sort all of the arguments in a unique way. 

% search of a suitable \textit{standard ordering} among the predicate arguments. Indeed, in the applications cited above, having a routine that permutes a predicate's arguments according to some kind of (unique) order is one straightforward way of automatically deriving a mapping among arguments (the mapping based on the argument's order). 

In this work, we introduce the notion of an argument profile being an abstract characterization of how that argument is used within the predicate and we devise an analysis capable of computing such profiles. Our approach encompasses, to some extent, type and mode information, but goes further by incorporating into the abstract domain the operations in which the argument participates. While the result of our analysis is not guaranteed to identify each and every argument by a unique value, examples show that it is capable of distinguishing between arguments much more precisely than approaches using only type and mode information.

%%%%%%%%%%%%%%%%%%%%%%%%%%%%%%%%%%%%%%%%%%%%%%%%%%%
%%%%%%%%%%%%%%%%%%%%%%%%%%%%%%%%%%%%%%%% 2
%%%%%%%%%%%%%%%%%%%%%%%%%%%%%%%%%%%%%%%%%%%%%%%%%%%

\section{Basic Concepts and Notations}\label{sec:1}
In this paper we consider a simple logic language $\mathcal{L}$ where predicates, clauses, atoms and terms are used and defined in a style similar to that of Prolog. The language is however moded and represents, as such, certain similarities with (a subset of) Mercury~\cite{mercury}.
We assume given a finite set of variables $\mathcal{V}$, a finite set of functor symbols $\mathcal{F}$ and a finite set of predicate symbols $\mathcal{P}$. As usual variables in $\mathcal{V}$ are strings starting with an uppercase letter while functors and predicates from $\mathcal{F}$, respectively $\mathcal{P}$ are written $p/n$ where $p$ is a string starting with a lowercase letter or symbol called the name of the functor (resp. predicate) and $n\in \mathbb{N}$ its arity, i.e. its number of arguments. We will ease notation by supposing that if a predicate (or functor) $p/n$ exists in the program, then no predicate (or functor) $p/m$ with $m \neq n$ can exist, so that a predicate (or functor) $p/m$ will sometimes simply be referred to as $p$. The set of terms constructed from $\mathcal{V}$ and $\mathcal{F}$ is denoted $\mathcal{T}$. A term $t\in\mathcal{T}$ is said to be ground if it contains no variables. 

A program is defined as a set of predicate definitions, where each predicate is defined by a set of clauses.
For simplicity, we will consider only definite clauses, that is each clause is of the form $H\leftarrow B_1,\ldots,B_n$ where $H$ is an atom denoted the head of the clause, and $B_1,\ldots,B_n$ a conjunction of atoms denoting its body. We furthermore assume that the head of a clause contains only variables as arguments (all unifications are made explicit in the body) and that all clauses defining a predicate share the same head. For a predicate $p$ we will use $\pred(p)$ to denote the set of clauses in its definition and $\args(p)$ to denote the sequence of its formal argument variables. With a slight abuse of notation we denote by $\args(p)_i$  the $i$th formal argument of $p$ ($i$ being a number between 1 and the arity of $p$). For any given program construction $c$, be it a predicate, a clause, an atom or a clause head, we denote by $\vars(c)$ the set of variables occurring in $c$. We will suppose that each atom in the program is uniquely identified by a natural number from $\mathbb{N}$ that will be referred to as the atom's program point in the program.

We will restrict ourselves to programs that are \textit{directly recursive} to ease the analysis formulation and obtain concrete and efficient results~\cite{efficient-dataflow}. Without loss of generality, we will also assume that clause bodies are in some standard, flattened, form in which each atom is either a predicate call having only variables as arguments, or a unification between variables and/or terms in which each term has only an outermost functor (its arguments being variables). 
We consider our language to be \textit{moded}: each argument appearing in a clause's head is characterized as being either input or output. The argument modes restrict the usage of the predicate in the sense that any call to the predicate must provide a fully instantiated (ground) value for the input arguments, whereas each output argument will be a free variable that is guaranteed to be bound to a ground value upon success of the call. Likewise, unifications are moded as well. 

\begin{definition}
	\label{def:unification}
	A moded unification is an atom in one of the following forms. 
	\begin{itemize}
		\item $V \Rightarrow f(X_1,\ldots, X_n)$, called \textit{deconstruction}, where $V$ is supposed to be input and $X_1,\ldots,X_n$ output. It succeeds if the value bound to $V$ has $f/n$ as an outermost functor in which case it binds $X_1,\ldots,X_n$ to the values figuring in the arguments of $f/n$.
		\item $V \Leftarrow  f(X_1,\ldots, X_n)$, called \textit{construction}, where $V$ is supposed to be output and $X_1,\ldots,X_n$ input. The construction succeeds if during evaluation $f(X_1,\ldots, X_n)$ is a ground value that can be bound to the free variable $V$.
		\item $V \leftrightarrow W$, called \textit{test}, where both $V$ and $W$ are supposed to be input. The test succeeds if both $V$ and $W$ are bound to identical ground values.
		\item $V := W$, called \textit{assignment}, where $V$ is supposed to be output, and $W$ input. The assignment succeeds if $W$ is bound to a ground value that can be assigned to the free variable $V$.
	\end{itemize} 
\end{definition}

Given these constructions and the moded context, our predicates do to some extent resemble what are called \textit{procedures} in Mercury~\cite{mercury}. 

\begin{example}
	\label{ex:append}
	If we represent lists in the usual way, by a functor $\mathit{nil}$ representing the empty list and a functor $\mathit{cons/2}$ for list construction, the predicate \texttt{app/3} below, to be used in a mode \texttt{(input,input,output)} realizes the classical ground list concatenation operation in $\mathcal{L}$. The first two arguments are thus supposed to be input, the third one output. The subscript numbers represent the atoms' program points. 
	\vspace{-0.1cm}
		\[\begin{array}{lll}
			app(X,Y,Z) & \leftarrow & X \Rightarrow_1 nil, Z :=_2 Y. \\
			app(X,Y,Z) & \leftarrow & X \Rightarrow_3 cons(E, Es), app_4(Es, Y, Zs), Z \Leftarrow_5 cons(E,Zs).\\ 
		\end{array}\]
	\vspace{-0.3cm}
		%\caption{The  predicate \texttt{app/3}, defined to be used in a mode \texttt{app(in,in,out)}}
\end{example}

In the remainder of the paper, we will use $\mathcal{A}$ to represent the set of atoms (predicate calls and unifications) as they can occur in the program text, i.e. in the flat form defined above. For an atom $A\in\mathcal{A}$, we denote by $\inargs(A)$ the {input arguments of $A$ and by $\outargs(A)$ its output arguments. Note that this only concerns variables, i.e. for any $A\in\mathcal{A}$ we have $\inargs(A)\subseteq\vars(A)$ and $\outargs(A)\subseteq\vars(A)$.
As usual, a substitution is a mapping from variables to terms and applying a substitution $\theta$ to a syntactical construct $e$, written $e\theta$, denotes the construct obtained by simultaneously replacing in $e$ all variables from the domain of $\theta$, denoted $\dom(\theta)$, with their corresponding value. Given substitutions $\theta$ and $\sigma$, their composition $\theta\circ\sigma$ is also written as $\theta\sigma$. A \textit{renaming} $\rho : \mathcal{V}\mapsto \mathcal{V}$ is a special kind of substitution as it is an injective (and idempotent) mapping between variables. 

We suppose that programs, when executed, behave in a mode-correct way, meaning that if an instance of an atom (be it a unification or a predicate call) is selected for resolution, the arguments in the atom's input positions are bound to ground values, whereas the arguments in the output positions are unbound variables. To formalize the semantics of our language, we thus introduce the notion of a mode-correct instance.

\begin{definition}
Let $A\in\mathcal{A}$ be an atom (predicate call or unification). We say that $A'$ is a \textit{mode-correct instance} of $A$ if and only if there exists a substitution $\theta$ such that $A'=A\theta$ and 
\begin{enumerate}
	\item[(1)] $\forall X\in\inargs(A) : \theta(X)$ is a ground term;
	\item[(2)] $\forall X\in\outargs(A): \theta(X)$ is a free variable if $X\in\dom(\theta)$.
\end{enumerate}
\end{definition}
 
 The semantics of the moded unifications defined above can easily be defined as follows:
 
 \begin{definition}
 Let $U\in\mathcal{A}$ denote a unification and $U\theta$ (for some substitution $\theta$) a mode-correct instance. Then we say that $U\theta$ \emph{succeeds with answer} $\theta'$ if and only if the following holds:
 \begin{itemize}
 \item If $U$ is of the form $X\Rightarrow f(Y_1,\ldots, Y_n)$ it holds that $\theta(X) = f(t_1,\ldots,t_n)$ and $\theta'=\{Y_1/t_1,\ldots,Y_n/t_n\}$.
 \item If $U$ is of the form $X\Leftarrow f(Y_1,\ldots,Y_n)$ it holds that $\theta'=\{X/f(\theta(Y_1),\ldots,\theta(Y_n))\}$.
 \item If $U$ is of the form $X \leftrightarrow Y$ it holds that $\theta(X)=\theta(Y)$ and $\theta'=\emptyset$.
 \item If $U$ is of the form $X:=Y$ it holds that $\theta'=\{X/\theta(Y)\}$.
 \end{itemize}
 %Otherwise, we say that  $U\theta$ \emph{fails}.
 \end{definition}
 
 The operational semantics of a program is defined in function of a query as usual.

\newcommand{\mgu}{\mathit{mgu}}
 \begin{definition}
 Given a program $P$, let $Q$ be a query of the form $\leftarrow A_1,\ldots,A_n$. We say that a query $Q'$ \emph{is derived from} $Q$ \emph{with answer} $\theta$ if and only if one of the following conditions holds:
 \begin{enumerate}
 \item $A_1$ is a mode-correct instance of a unification that succeeds with answer $\theta$, and $Q'$ is the query $\leftarrow (A_2,\ldots,A_n)\theta$.
 \item $A_1$ is a mode-correct instance $p(t_1,\ldots,t_n)$ of the head $H=p(X_1,\ldots,X_n)$ of a (renamed apart) clause $H\leftarrow B_1,\ldots,B_k\in P$ and $Q'$ is the query $\leftarrow (B_1,\ldots, B_k,A_2,\ldots, A_n)\theta$ where it holds that $\theta=\{X_1/t_1,\ldots,X_n/t_n\}$.
 \end{enumerate}
 %If none of the conditions above is satisfied, we say that the query $Q$ fails in $P$.
 \end{definition}
 
 The above definition is basically equivalent to a traditional SLD-resolution step (with a leftmost selection rule) except for the explicit handling of the (moded) unifications and the limitation to resolving mode-correct instances of atoms only. Next, we can define the notion of a derivation as a sequence of individual derivation steps.
 
 \begin{definition}
Given a program $P$ and query $Q_0$. A \emph{derivation} for $Q$ in $P$ is a sequence of queries and substitutions $Q_0\stackrel{\theta_0}{\rightarrow} Q_1\stackrel{\theta_1}{\rightarrow}\ldots\stackrel{\theta_{n-1}}{\rightarrow}Q_n$ such that $Q_i$ is derived from $Q_{i-1} $ with answer $\theta_{i-1}$ for each $1\le i\le n$. If $Q_n$ is the empty query $\diamond$ then we say that the derivation is \emph{successful} and has associated computed answer substitution $\theta_0\theta_1\ldots \theta_{n-1}$.
\end{definition}

Again, our notion of a derivation is essentially equivalent to an SLD-derivation with a left-to-right selection rule. However, as a consequence of the simple mode system, all computed answers are ground substitutions.

\section{Argument and Predicate Profiles}\label{sec:2}
The analysis described in the next section essentially interprets a well-moded logic program and registers the encountered operations into special sets called \textit{interaction sets} that will in the end allow to define a so-called \textit{profile} for each of the predicate's arguments. %Intuitively, an argument's profile is composed by the operations by which this arguments potentially constructs, or is constructed by, another argument in the predicate's execution.
The key idea of this section is to formalize the values that will be computed and manipulated by our analysis. 

First, let us abstract $n$-ary computations by the dataflow relations that are exhibited between the arguments of a predicate, each dataflow relation being annotated by the set of operations that participate in the relation. Among the operations of interest are the basic unification operators defined by the set $B$ as follows:
\[B =  \{:=, \leftrightarrow\} \cup \bigcup_{f\in\mathcal{F}} \{\Leftarrow_f,\Rightarrow_f\} \]
For a given argument, we will represent a single dataflow relation it participates in by means of an \emph{o-set}, the latter being essentially a tuple $(o,j)$ in which $o$ represents a subset of operations (from a given set of admissible operations, like $B$ above) and $j$ a natural number representing the position of one of the (other) arguments. More formally:

\begin{definition}
	Given a set of operations $S$, we define the set of \emph{o-sets over $S$} as
	\[\OS(S) = \{(o,j)\:|\:o\in\pset(S)\mbox{ and }j\in\mathbb{N}\}\]
\end{definition}

In general, an argument participates in more than one dataflow relation, relating it to several other arguments (each time by means of a set of operations). To represent such a \textit{set} of dataflow relations, we introduce the notion of an \textit{argument profile}.
Intuitively, an argument profile for the $i$'th argument of $p/n$ denotes a set of dataflow dependencies with some of the other arguments of $p$, where each dependency is represented -- through an $o$-set -- by the set of operations linking both arguments. Formally, we define the notion of an argument profile for an $n$-ary operation as follows:

\begin{definition}
	Given a set of operations $S$ and $n\in\mathbb{N}$, 
	we define an argument profile for an $n$-ary operation with respect to $S$ as a set $A\subseteq\OS(S)$ where for each $(o,j)\in A$ we have that $j\in\{1,\ldots,n\}$.
	We will use $\AP_n(S)$ to represent the set of all possible argument profiles for an $n$-ary operation with respect to $S$.
\end{definition}

\begin{example}
	The following is an argument profile: 
	$\{(\{\Rightarrow_{cons}, :=\}, 2), (\{\Leftarrow_{cons}\}, 3)\}$. 
	It represents the fact that the concerned argument is involved through a deconstruction in a list, and an assignment, with the value of the argument in position 2. It also helps building the argument in position 3 through a list construction atom. 
\end{example}

The above definitions are fine as long as we restrict ourselves to using operations from a fixed set of operations such as $B$. However, it is worthwhile to include among the allowed operations also those operations defined (by means of predicates) in the program itself.
We will not include the predicates as such in the set of admissible operations as it would make the domain too dependent on the names chosen for the predicates at hand. Rather, we will use abstractions of these predicates -- notably those abstractions our analysis aims to compute. 
As such, the basic idea is to represent an $n$-ary operation (or predicate) by means of a term $\psi(\alpha_1,\ldots, \alpha_n)$ where the $\alpha$ are argument profiles. A special term $\psi_\bot$ is introduced in order to represent an operation for which no argument profiles are known; in the analysis it will be used to represent direct recursive calls.
Since these $\psi$-based terms use argument profiles that themselves can contain $\psi$-based terms, we define the set of all possible abstract operations as the least fixed point of the following operator $R$:
\begin{definition}
	Given a set of operations $S$, we define 
	\vspace{-0.05cm}
	\[R(S) = B\cup \{\psi_\bot\}\cup\bigcup_{n\in\mathbb{N}_0}\{\psi(\alpha_1,\ldots,\alpha_n)\:|\:\alpha_i\in\AP_n(S) \}\]
\end{definition}
\vspace{-0.15cm}
While $\lfp(R)$ contains some infinite terms, all terms created by our analysis will be of finite size, as will be made clear further down. In the following we use $\AP_n$ to refer to the set of all possible argument profiles for an $n$-ary operation with respect to $\OS(\lfp(R))$. We will refer to the elements of $\lfp(R)$ in which a $\psi$ appears as \textit{$\psi$-based operations}. 
 %we will thus consider the set $\OS(\lfp(R))$ as the o-sets of interest and refer to this set simply by $\OS$. Likewise, . \wim{Does this makes sense to avoid a too heavy notation?}

In order to obtain argument profiles, the analysis will compute data flow relations within a predicate, annotated with the operations that are encountered upon establishing the relation. We thus define an \textit{interaction} as being the association of an input variable and an output variable with a set of operations and the program points these operations are occurring at. Formally:

 \begin{definition}
	Let $p$ be a predicate in a program $P$. An \emph{interaction in $p$} is a mapping $\vars(p) \times \vars(p) \mapsto \pset(\lfp(R)\times\mathbb{N})$. Notation-wise, we will typically write $V \xrsquigarrow{O}\hat{V}$ to represent an interaction between a variable $V$ and another variable $\hat{V}$ through a set $O\subset \lfp(R)\times\mathbb{N}$. %The set of all interactions is denoted $\Theta$. 
\end{definition}

In order not to overload our notation, when writing interactions, we will usually drop the program points and consider the sets of operations in an interaction to be a multiset $O\subset\lfp(R)$. We will thus allow doubles in the set, assuming they are operations implemented by atoms located at different program points. We will only occasionally include program points explicitly when needed in order to explicitly distinguish between identical operations coming from different atoms.

An important characteristic of the set of interactions describing a predicate is that for each pair of variables, there is at most a single interaction between these variables present in the set. Another characteristic is the fact that for any interaction $V \xrsquigarrow{O}\hat{V}$ it holds that $\hat{V}$ cannot be an input argument, since mode-correct input arguments cannot be constructed by computations in a predicate's body. $V$ does not have such a limitation, as long as $V$ and $\hat{V}$ are distinct. More formally:

\begin{definition}
For a predicate $p$, we call a \textit{well-defined interaction set} for $p$ a set $\phi$ of interactions in $p$ such that for all $V,\hat{V}\in\vars(p)$ it holds that if there exists $V\xrsquigarrow{O}\hat{V}\in\phi$ for some $O$, then the following conditions all hold: 
\begin{enumerate}
	\item $V \neq \hat{V}$;
	\item $\nexists V\xrsquigarrow{O'}\hat{V}\in\phi : O'\not=O$;
	\item $\hat{V}\in\args(p) \Rightarrow \hat{V}$ is an output argument.
\end{enumerate}
\end{definition}

We will use $\iset_p$ to denote the set of all well-defined interaction sets for a given predicate $p$. In case $p$ is clear from the context, we will use the shorter notation $\iset$. Now we define the following quasi-order allowing to organize $\iset_p$ in a lattice. 
\begin{definition}
	Let $p$ be a predicate. For $\phi_1, \phi_2 \in \iset_p$ we say that $\phi_1$ is more precise than $\phi_2$, denoted $\phi_1\sqsubseteq\phi_2$, if and only if $\forall V\xrsquigarrow{O}\hat{V} \in \phi_1 : \exists V\xrsquigarrow{O'}\hat{V} \in \phi_2$ such that $O\subseteq O'$. 
\end{definition}

That is, $\phi_1\sqsubseteq\phi_2$ when each interaction appearing in $\phi_1$ labeled by an operation set $O$ is matched by an interaction in $\phi_2$ that is labeled by an operation set being a superset of $O$, and $\phi_2$ may contain interactions involving pairs of variables that are not linked by an interaction in $\phi_1$. 
We now define the following operator.

\begin{definition}
For a predicate $p$, let $\phi\in\iset_p$ and let $V\xrsquigarrow{O}\hat{V}$ be an interaction for $p$. Then we define
\[(V\xrsquigarrow{O}\hat{V})\: \sqcup\: \phi = \left\{\begin{array}{lll}
\{V\xrsquigarrow{O}\hat{V}\}\cup \phi & & \mbox{if } \nexists (V\xrsquigarrow{O'}\hat{V})\in\phi\mbox{ for some } O'\\
(\phi\setminus \{V\xrsquigarrow{O'}\hat{V}\})\cup\{V\xrsquigarrow{O\cup O'}\hat{V}\} & & \mbox{otherwise}
\end{array}\right.\]
\end{definition}

Note that adding an interaction to a well-defined interaction set results in a well-defined interaction set. It can also be easily seen that when constructing a well-defined interaction set, the order in which the individual interactions are added has no influence on the final result. Consequently, we can extend the $\sqcup$ operator such that it merges two well-defined interaction sets:

\begin{definition}
Let $\phi$ and $\phi'$ be well-defined interaction sets for a predicate $p$. Then we define $\phi\sqcup\phi'$ as the following well-defined interaction set: $\phi\sqcup\phi' =\bigsqcup_{V\xrsquigarrow{O}\hat{V}\in\phi} (V\xrsquigarrow{O}\hat{V})\sqcup \phi'$.
\end{definition}

\begin{proposition}\label{prop:semi-lattice}
For a predicate $p$, $(\iset_p,\sqcup)$ is a join semi-lattice.
\end{proposition}
\begin{proof}
	We need to prove that for a predicate $p$, the 
	$\sqcup: \iset_p\times\iset_p\mapsto\iset_p$ operation is idempotent, associative and commutative. This follows directly from the definition of $\sqcup$ (being essentially a union operation on sets of interactions and possibly on sets of operations) and the fact that the union operator on sets is itself idempotent, associative, and commutative.
\end{proof}

The induced partial order, namely $\sqsubseteq$, is such that $\phi\sqsubseteq\phi'$ if and only if $\phi\sqcup\phi'=\phi'$, so that we get a partially ordered set $(\iset_p,\sqsubseteq)$ in which each subset $\{\phi_1,\ldots,\phi_n\}$ has a least upper bound, namely $\sqcup \{\phi_1,\ldots,\phi_n\}$. The partially ordered set has a minimal element, namely the empty set $\{\}$ which we will refer to by $\bot$ as it is a unit for the join operator: $\forall \phi\in\iset_p : \bot\sqcup\phi = \phi\sqcup\bot=\phi$. The maximal element $\top_p \in \iset_p$ is the set containing all possible interactions between each argument and all the (other) output arguments.  

The goal of our analysis is to compute, for each predicate $p$ in a given program $P$, a well-defined interaction set for $p$. This element of $\iset_p$ will be such that it only reflects the interactions between variables $V, \hat{V}$ such that $V, \hat{V} \in \args(p)$. Such an element is what we will call a \textit{predicate profile}.

\begin{definition}
	Given a program $P$ and a predicate $p$ defined therein. A \emph{predicate profile} for $p$ is a well-defined interaction set $\phi$ of interactions in $p$ such that for all $V\xrsquigarrow{O'}\hat{V}\in\phi$ we have that $V$ and $\hat{V}$ are formal arguments of $p$, that is $\{V,\hat{V}\}\subseteq\args(p)$.
\end{definition}

We can "decompose" a predicate profile into individual argument profiles as follows:

\begin{definition}
	Given a program $P$, a predicate $p$ in $P$, and a predicate profile $\phi$ for $p$, we define the argument profile of the $i$'th argument of $p$ with respect to $\phi$ as the following set of o-sets:
	\vspace{-0.1cm}
	\[\alpha_i = \{(O, j)\:|\:V_i\xrsquigarrow{O}V_j\in\phi\}\]
	where $V_i = \args(p)_i$ and $V_j=\args(p)_j$. Moreover, we define the \emph{computed argument profile} of $p$ with respect to $\phi$ as the sequence $\langle \alpha_1, \ldots, \alpha_n\rangle$.
\end{definition}

Recall that, based on such computed argument profiles, our objective is to \textit{reorder} the predicate arguments, preferably in a unique way. As a first observation, note that it is not hard to define \textit{a} total order on $\AP$ as the following example illustrates.

\begin{example}\label{ex:sort}
	For an argument profile $\alpha \in\AP$, let us define the \emph{features of $\alpha$} as the vector $(\#\alpha, o, m, s, r, c, d)$ with $o$ the total number of operations contained in $\alpha$, $r$ the number of $\psi$-based operations in it, $c$, $a$, $d$ its number of constructions, assignments and deconstructions respectively. 
	Denoting by $(0)$ a vector filled with zeroes, we define the total order $\le$ as the operator such that for any two argument profiles $\alpha_1$ and $\alpha_2$ with respective features $t_1$ and $t_2$, the following holds:  
	\begin{center} $\alpha_1 \le \alpha_2 \Leftrightarrow t_1 - t_2 = (0) \vee \mbox{the first non-zero dimension in }t_1 - t_2\mbox{ is positive}$\end{center}
\end{example}

While the order of Example~\ref{ex:sort} is somewhat arbitrary and not necessarily capable of producing a \textit{unique} order, its definition is independent of the analyzed program. In the following section, we construct our analysis that takes a total order $\le$ on $\AP$ as a parameter. Given such an order $\le$, for a predicate $p$ with some profile $\phi$, we will use $\oprof(\phi)$ to represent a profile of $p$ ordered by $\le$ with respect to $\phi$.

\begin{definition}
	Given a predicate $p/n$, a profile $\phi$ and a total order $\le$. Let $\langle \alpha_1, \ldots, \alpha_n\rangle$ be the argument profile of $p$ with respect to $\phi$. Then we define the \emph{ordered profile} of $p$ with respect to $\phi$ as a permutation $\langle \alpha_1',\ldots,\alpha_n'\rangle$ of $\langle \alpha_1, \ldots, \alpha_n\rangle$ such that $\alpha_i\le\alpha_{i+1}$ for all $1\le i<n$.
\end{definition}

%%%%%%%%%%%%%%%%%%%%%%%%%%%%%%%%%%%%%%%%%%%%%%%%%%%
%%%%%%%%%%%%%%%%%%%%%%%%%%%%%%%%%%%%%%%% 4
%%%%%%%%%%%%%%%%%%%%%%%%%%%%%%%%%%%%%%%%%%%%%%%%%%%

%
\section{A Dataflow Analysis Computing Argument Profiles}\label{sec:analysis}
The analysis will basically compute what we call an \textit{environment} which is a mapping from predicates to well-defined interaction sets that represent the already computed interactions between the predicate's formal arguments. We will use the symbol $\Phi : \mathcal{P}\mapsto \iset$ to represent such an environment. %The initial environment, $\Phi_0$ will be defined such that $\Phi(p)=\bot$ for each $p\in\mathcal{P}$.\gonz{dit plus loin}
%We can now define the analysis itself. %, which essentially computes, for a predicate $p$, the set of interactions that represent the dataflow relations between $p$'s arguments. 
The analysis is defined by induction on the syntactic structure of the program's predicates. We start by defining the analysis of an individual atom. It basically incorporates the operations of interest into interactions involving local variables as well as arguments. The analysis is parametrized by the current environment $\Phi$ and a total order $\le$ capable of ordering a predicate profile $\phi$ into $\oprof(\phi)$.

\begin{definition}\label{def:a}
	Let $P$ be a program of interest. The atomic analysis function $\mathbb{A} : \mathcal{A} \mapsto (\mathcal{P}\mapsto \iset) \mapsto \iset$ is defined as the function that returns, given an atom $A$ and an environment $\Phi$, a set of interactions composed by those operations from $\lfp(R)$ that are found occurring in $A$:	
	\vspace{-0.1cm}	
	\begin{align*}
		& \mathbb{A}\llbracket V \Rightarrow f(Y_1,\dots,Y_n) \rrbracket \Phi 
	 	= \underset{i\in1..n}{\bigsqcup}\{V \xrsquigarrow{\{\Rightarrow_f\}}Y_i\} \\
		&\mathbb{A}\llbracket V \Leftarrow f(Y_1,\dots, Y_n) \rrbracket \Phi
		 = \underset{i \in 1..n}{\bigsqcup}\{Y_i\xrsquigarrow{\{\Leftarrow_f\}}V\}  
		\\
		& \qquad \qquad \; \: \mathbb{A}\llbracket V := W\rrbracket \Phi = \{W\xrsquigarrow{\{:=\}}V\} \\	
		&\qquad \qquad \quad \; \; \mathbb{A}\llbracket V \leftrightarrow W \rrbracket \Phi = \{\}\\
		&\qquad \; \; \; \; \mathbb{A}\llbracket q(Y_1,\dots,Y_m)\rrbracket  \Phi = \Phi(q)\rho \sqcup \phi_q 
	\end{align*} 
%	\[\begin{array}{rlll}
%	 \mbox{where} & \rho & = & \{\args(q)_1/Y_1, \dots, \args(q)_m/Y_m\} \\
%	 \mbox{and}& \phi_q & = & \left\{Y_i\xrsquigarrow{o}Y_j\:|\: Y_i\in\inargs(q(Y_1,\dots,Y_m)), Y_j\in\outargs(q(Y_1,\dots,Y_m))\right\}\\
%	\mbox{in which}	& o & =  & \left\{\begin{array}{ll}
%			\psi_\bot & \mbox{if it is a directly recursive call}\\
%			\psi(\oprof(\Phi(p))) & \mbox{otherwise}
%		\end{array}\right.
%	\end{array}\]
	%
	\vspace{-0.6cm}
	\begin{align*}
		& \mbox{where } \rho = \{\args(q)_1/Y_1, \dots, \args(q)_m/Y_m\} \\
		& \mbox{  and  } \phi_q = \left\{Y_i\xrsquigarrow{o}Y_j\:|\: Y_i\in\inargs(q(Y_1,\dots,Y_m)), Y_j\in\outargs(q(Y_1,\dots,Y_m))\right\}\\
		& \qquad \mbox{in which } o =  \left\{\begin{array}{ll}
			\psi_\bot & \mbox{if it is a directly recursive call}\\
			\psi(\oprof(\Phi(p))) & \mbox{otherwise}
		\end{array}\right.
	\end{align*}
\end{definition}

In the definition, we apply a renaming $\rho$ to a set of interactions $\Phi(q)$, which consists in replacing each variable $V$ from $\dom(\rho)$ occurring in $\Phi(q)$ by $\rho(V)$. Using $\oprof(\Phi(p))$ allows the $\psi$-based operations occurring in an argument profile to describe atoms based on similar operations by means of normalized values. For instance, as will be made clear later on, whether a predicate makes a call to $app/3$ or to a variant of it where some arguments are swapped, the resulting $\psi$-based operation will be the same.
\begin{example}\label{ex:analysis1}
	%%\[\begin{array}{lll}
		%	app(X,Y,Z) & \leftarrow & X \Rightarrow_1 nil, Z :=_2 Y. \\
		%	app(X,Y,Z) & \leftarrow & X \Rightarrow_3 cons(E, Es), app_4(Es, Y, Zs), Z \Leftarrow_5 cons(E,Zs).\\ 
		%\end{array}
		%\]
		The following are applications of our function $\mathbb{A}$ on atoms that appear in the predicate $app$ from Example~\ref{ex:append}. We consider given an environment $\Phi_0$ that maps $app$ on $\bot$. 
			\vspace{-0.15cm}
		\begin{align*}
			& \mathbb{A}\llbracket X \Rightarrow cons(E,Es) \rrbracket \Phi_0 = \{X \xrsquigarrow{\{\Rightarrow_{cons} \} } E, X \xrsquigarrow{\{\Rightarrow_{cons}\}} Es\} \\
			& \, \mathbb{A}\llbracket Z \Leftarrow cons(E,Zs) \rrbracket \Phi_0 = \{E \xrsquigarrow{\{\Leftarrow_{cons} \} } Z, Zs \xrsquigarrow{\{\Leftarrow_{cons}\}} Z\} \\
			& \quad \; \, \mathbb{A}\llbracket app(Es,Y,Zs) \rrbracket \Phi_0 =
			\{Es \xrsquigarrow{\{\psi_\bot \} } Zs, Y\xrsquigarrow{\{\psi_\bot\}} Zs\}
		\end{align*}
\end{example}

Extending the analysis function to clauses is relatively straightforward as it suffices to analyze each of the body atoms, joining the results using $\sqcup$. However, we need to include a transitive closure operator that allows to \textit{combine} the interactions resulting from the analysis of the individual atoms such that the resulting interactions represent~-- where possible -- data flow between arguments rather than involving local variables. 

%\gonz{ Dans la définition, j'ai ajouté "distinct" pour les trois variables pour balayer le cas où on aurait A$\rightarrow$B$\rightarrow$A. Mais je crois que ce cas est impossible si on a le mode-correctness qui est garanti. A vérifier !} 
\begin{definition}
	Let $p\in\mathcal{P}$ and $\phi \in \iset_p$.
	Let $T:\iset\mapsto\iset$ denote the following operator
		\vspace{-0.1cm}
		\[T(\phi) = \{ X\xrsquigarrow{O\:\cup\: O'}Z\:|\:X\xrsquigarrow{O}Y, Y\xrsquigarrow{O'}Z \in \phi\mbox{ for some distinct }X,Y,Z\in\mathcal{V}\}\]
	and let $cl_T(\phi)$ denote the transitive closure of $T$ on $\phi$, that is the smallest relation on $\phi$ that contains $T$ and is transitive. Then the  
	\textit{projection of $\phi$ onto the arguments of $p$} is denoted by $\pi_p(\phi)$ and defined as
	\vspace{-0.1cm}
	\[\pi_p(\phi) = \{ X\xrsquigarrow{O}Y\in cl_T(\phi)\:|\:X,Y\in\args(p)\}.\]
\end{definition}

%\gonz{Le § suivant est-il compréhensible ?}  
For a given $\phi\in\iset$, the transitive closure $cl_T(\phi)$ can always be computed by merging into $\phi$ those interactions that can be seen as \textit{transitive interactions}, i.e. interactions that concern three different variables $X, Y, Z$ in the way described in the Definition above. The number of these transitive interactions is inevitably finite, being proportional to the number of combinations among a finite number of variables.

The analysis of a complete program consists in repeatedly analyzing each and every clause of the program with respect to the current environment, computing as such an updated environment that incorporates the results of the current analysis round. 

\begin{definition}
	Let $P$ be a program and $p\in P$ a predicate of interest. The predicate analysis function $\mathbb{S} : \mathcal{P} \mapsto (\mathcal{P}\mapsto \iset) \mapsto \iset$ is defined as the function that returns, given a predicate $p$ and an environment $\Phi$, a well-defined interaction set for $p$:
	\[\mathbb{S}\llbracket p\rrbracket \Phi = \bigsqcup_{h\leftarrow a_1,\ldots,a_n\in \pred(p)}
	\pi_p(\bigsqcup_{i\in 1\ldots n} \mathbb{A}\llbracket a_i\rrbracket \Phi)\] 
	%where $\pi_p$ represents the  function projecting the computed interaction set onto the formal arguments of $p$.
\end{definition}

Note the effect of the different join operations. First, the interaction sets resulting from the analysis of the individual atoms in a clause body are combined (using the innermost join). The outermost join combines the interaction sets resulting from the different clauses, after projection, into a single interaction set. The projection onto the arguments of the predicate is important, as it avoids the construction of spurious interactions caused by the same local variable that might be used in different clauses. The fact that local variables are ignored in the result of the formula above is no limitation, since the operator $\mathbb{S}$ is used below to compute the successive environments, and our analysis uses the environment solely for exploiting the interactions among arguments. 

\begin{example}\label{ex:analysis2}
Let us consider again the predicate $app$ from Example~\ref{ex:append}. A round of our analysis for $app$ is partially depicted in Example~\ref{ex:analysis1}, its result being $\mathbb{S}\llbracket app \rrbracket \Phi_0   = 
	\{Y\xrsquigarrow{\{:=, \psi_\bot\}}Z, X \xrsquigarrow{\{\Rightarrow_{cons}, \psi_\bot \Leftarrow_{cons}\}}Z \}$,
which corresponds to the projection on $X$, $Y$ and $Z$ of the following computed interactions: 
\begin{center}$ \{Y \xrsquigarrow{\{:=\}}Z, X\xrsquigarrow{\{\Rightarrow_{cons}\}}E, X\xrsquigarrow{\{\Rightarrow_{cons}\}}Es,
X\xrsquigarrow{\{\psi_\bot\}}Z,
Y\xrsquigarrow{\{\psi_\bot\}}Z,
E\xrsquigarrow{\{\Leftarrow_{cons}\}}Z, Zs\xrsquigarrow{\{\Leftarrow_{cons}\}}Z\}$\end{center}
\end{example}

Now, to analyze a program from scratch, we start from an initial environment $\Phi_0$ in which each predicate is associated to an initial interaction set $\bot$. The predicates are subsequently analyzed according to their position in the program's call graph in a bottom-up manner, that is prioritizing those predicates that contain no calls to predicates except maybe themselves or predicates that have previously been analyzed. 
We will denote by $\leafs(P)$ the set of such \textit{eligible} predicates in a program $P$. Each time a predicate's analysis reaches a fixpoint, the analysis proceeds to the next eligible predicate. The process is repeated until every predicate has been considered. It is depicted in Algorithm~\ref{alg:env}.
%
%\begin{algorithm}[hbtp]
%	\caption{Analyzing a program $P$}
%	\label{alg:env1}
%	\begin{algorithmic}
%		\State $PS \gets P, \Phi \gets \bigcup_{p\in P} \{\:(p, \bot)\:\}$ 
%		\While{$\leafs(PS) \neq \emptyset$}
%			\State select $p\in\leafs(PS)$
%			\While {$\Phi(p) \neq \mathcal{S}\llbracket p \rrbracket \Phi$}
%				\State $\Phi(p) \gets \mathcal{S}\llbracket p \rrbracket \Phi$
%			\EndWhile
%			\State $PS \gets PS \setminus \{p\}$
%		\EndWhile
%	\end{algorithmic}
%\end{algorithm}	
%
%\gonz{Deuxième formulation : }
%
\begin{algorithm}[tbp]
	\caption{Analyzing a program $P$}
	\label{alg:env}
	\begin{algorithmic}
		\small 
		\State $PS \gets P, i \gets 0, \Phi_0 \gets \bigcup_{p\in P} \{\:(p, \bot)\:\}$ 
		\While{$\leafs(PS) \neq \emptyset$}
			\State select $p\in\leafs(PS)$
			\While{$(\mathbb{S}\llbracket p \rrbracket \Phi_i)(p) \neq \Phi_i(p)$}
				\State $\Phi_{i+1} \gets \mathbb{S}\llbracket p \rrbracket \Phi_i$ 
				\State $PS \gets PS \setminus \{p\}$
				\State $i \gets i + 1$
			\EndWhile
		\EndWhile
	\end{algorithmic}
\end{algorithm}
%\begin{definition}\label{def:phi}
%	Let $P$ be a program and $Z = \{\}$. We define a sequence of environments as follows:
%	\[\begin{array}{lllll}
%		\Phi_0 & = & \bigcup_{p\in P} \{\:(p, \bot)\:\}\\
%		\Phi_{n+1} & = & \bigcup_{p\in P} \{\:(p, \mathcal{S}\llbracket p \rrbracket \Phi_n)\:\} & & \forall n\in\mathbb{N}\\
%	\end{array}\]
%\end{definition}
%In what follows we will use $\pa(P)$ to refer to the limit of the sequence $(\Phi_n)$ as defined above for a program $P$. As such, $\pa(P)$ represents the result of our analysis for $P$, i.e. a mapping containing, associated to each predicate, its computed profile. 
\begin{example}\label{ex:analysis3}
	Let us resume the analysis of $app/3$ started in Examples~\ref{ex:analysis1} and~\ref{ex:analysis2}, where we obtained an environment value, say $\Phi_1$, after one analysis round. A second round
	of the analysis will only differ in the handling of the atom $app(Es,Y,Zs)$: 
	\vspace{-0.1cm}
	\begin{center}
		$\mathbb{A}\llbracket app(Es,Y,Zs)\rrbracket\Phi_1 = \{Y\xrsquigarrow{\{:=, \psi_\bot \}}Zs, Es \xrsquigarrow{\{\Rightarrow_{cons}, \Leftarrow_{cons}, \psi_\bot\}}Z \} $
	\end{center}
	After merging and projection on the arguments, we obtain $\Phi_2$ such that
	\begin{center}
		$\Phi_2(app) = \{X \xrsquigarrow{\{\Rightarrow_{cons}, \Leftarrow_{cons}, \psi_\bot\}}Z, Y\xrsquigarrow{\{\Leftarrow_{cons}, :=, \psi_\bot\}}Z \}$
	\end{center}
	where the $\Leftarrow_{cons}$ operation linking $Y$ to $Z$ is obtained by the fact that we have both $Y\xrsquigarrow{\{:=\}}Zs$ and $Zs\xrsquigarrow{\Leftarrow_{cons}}Z$ in the computed interactions set. Any subsequent analysis round would not alter this environment, so that the analysis is finished for \texttt{app}. 
	
	Let us now consider that our program is also constituted of a moded version of the $concat/3$ predicate introduced in Section~\ref{sec:intro}: 
	\vspace{-0.1cm}
	\[\begin{array}{lll}
		concat(A,B,C) & \leftarrow & B \Rightarrow_6 nil, A :=_7 C. \\
		concat(A,B,C) & \leftarrow & B \Rightarrow_8 cons(I, Is), concat_9(As, Is, C), A \Leftarrow_{10} cons(I,As).
	\end{array}\]
	\vspace{-0.25cm}
	
	Analyzing $concat$ yields the interactions $\{B \xrsquigarrow{\{\Rightarrow_{cons}, \Leftarrow_{cons}, \psi_\bot\}}A, C\xrsquigarrow{\{\Leftarrow_{cons}, :=, \psi_\bot\}}A \}$.
	Now using $\le$, the ordered profiles of both predicates are one and the same, namely
	\begin{center}$\left\langle\left\{(\{\Rightarrow_{cons}, \Leftarrow_{cons}, \psi_\bot\}, 2)\right\}, \left\{(\{:=, \psi_\bot, \Leftarrow_{cons}\}, 2)\right\}\right\rangle$\end{center}
	which corresponds to the respective profiles of $X/B$, $Y/C$ and $Z/A$. %Thus, the analysis correctly allows to identify the corresponding arguments in both predicates. 
	In other words, reordering the arguments according to $\le$ leaves $app$ untouched but transforms $concat(A,B,C)$ into $concat(B,C,A)$. 
\end{example}

The predicate calls in the example above being recursive calls, we introduce the following example to illustrate the case where a predicate makes calls to other predicates.

\begin{example}\label{ex:dapp}
	Let us extend Example~\ref{ex:analysis3} with the \textit{double append} operation embodied by $dapp/4$:
	\begin{center}
		$dapp(L1,L2,L3,L4) \leftarrow app_{11}(L1,L2,L12), concat_{12}(L4,L12,L3).$
	\end{center}
	The analysis finds the following final interaction set for $dapp$: 
	\[\left\{\begin{array}{l} 
		L1 \xrsquigarrow{\{\Rightarrow_{cons} (3), \Leftarrow_{cons} (5), \psi_\bot (4), \psi_{a} (11), \Rightarrow_{cons} (8), \Leftarrow_{cons} (10), \psi_\bot (9), \psi_{a} (12)\}} L4, \\
		L2 \xrsquigarrow{\{:= (2), \psi_\bot (4), \Leftarrow_{cons} (5), \psi_{a} (11), \Rightarrow_{cons} (8), \Leftarrow_{cons} (10), \psi_\bot (9), \psi_{a} (12)\}} L4, \\ 
		L3 \xrsquigarrow{\{:= (7), \psi_\bot (9), \Leftarrow_{cons} (10), \psi_{a} (12)\}} L4
		\end{array}\right\}
		\]
	 where $\psi_{a} = \psi(\{(\{\Rightarrow_{cons}, \Leftarrow_{cons}, \psi_\bot\}, 2)\}, \{(\{:=, \psi_\bot, \Leftarrow_{cons}\}, 2)\})$ and where the program points have been made explicit when applicable. 
\end{example}

The example shows that our analysis allows to entirely distinguish the four arguments of $dapp/4$, whereas type- and mode information alone would not have made a distinction among the first three arguments. Having these profiles for different arguments allows to order these by using an appropriate $\le$ operator and, hence, to match $dapp/4$ with predicates that implement the same functionality differently.

A prototype implementation of the analysis, taking into account more elaborate examples, has been implemented. It is available online as an open source project\footnote{The artifact code is available as a GitHub repository located at \url{https://github.com/Gounzy/PredArgs}.}. The tool is capable of reordering predicate arguments and displaying the computed profiles for a directly recursive CLP program given as input. 

We conclude this section with two important observations on the analysis described above. First, we show that the algorithm terminates. Next, we give an upper bound of its computational complexity. 

%\begin{example}
%	Let us extend Example~\ref{ex:analysis} with the \textit{double append} operation embodied by the following $dapp/4$ predicate:
%	\[\begin{array}{lll}
%		dapp(L1,L2,L3,L4) & \leftarrow & app(L1,L2,L12), concat(L12,L3,L4). \\
%	\end{array}
%	\]
%	where the arguments of $concat/3$ have been reordered according to $\le$. 
%	The analysis finds the following value for $\pa(P)(dapp)$: 
%	\[\left\{\begin{array}{l} 
%		L1 \xrsquigarrow{\{\Rightarrow_{cons}, \Leftarrow_{cons}, \psi_\bot, \psi_{a}, \psi_{a}\}} L4, \\
%		L2 \xrsquigarrow{\{:=, \psi_\bot, \Leftarrow_{cons}, \psi_{a}, \Rightarrow_{cons}, \Leftarrow_{cons}, \psi_\bot, \psi_{a}\}} L4, \\ 
%		L3 \xrsquigarrow{\{:=, \psi_\bot, \Leftarrow_{cons}, \psi_{a}\}} L4
%	\end{array}\right\}
%	\]
%	
%	\[\mbox{where }\psi_{a} = \psi((\{\Rightarrow_{cons}, \Leftarrow_{cons}, \psi_\bot\}, 2), (\{:=, \psi_\bot, \Leftarrow_{cons}\}, 2))\]
%\end{example}

%One important result is the guarantee of the analysis termination. 
%\begin{proposition}\label{prop:convergent1}
%	Algorithm~\ref{alg:env1} terminates.  
%\end{proposition}
%
%\gonz{ Ou alors }
\begin{proposition}\label{prop:convergent}
	The sequence $(\Phi_n)$ as defined by Algorithm~\ref{alg:env} is convergent.  
\end{proposition}
\begin{proof}
	First note that by construction, the sequence of computed environments $\Phi_0,\Phi_1,\dots$ is such that $\forall i \in \mathbb{N}_0$, either $\Phi_i = \Phi_{i-1}$ and then $\Phi_i$ is the fixpoint of the sequence, or there exists $p\in\mathcal{P}$ such that $\Phi_i(p) \neq \Phi_{i-1}(p)$. In that case, the only possibilities are that
	\begin{itemize}
		\item $\Phi_i(p) \supset \Phi_{i-1}(p)$, due to a new interaction being discovered during the iteration, and/or
		\item $\exists V\xrsquigarrow{O_1}\hat{V}\in\Phi_i(p), V\xrsquigarrow{O_2}\hat{V} \in \Phi_{i-1}(p) : O_1\neq O_2$. This can only happen if a new operation is added to an existing interaction, or if a $\psi$-based operation is replaced by a different $\psi$-based operation.
		% and if a $\psi$-based operation computed by the analysis is replaced by another.\wim{Que veut dire cette dernière phrase ?}
	\end{itemize}
	%These are the only two possibilities since one analysis round differs from the previous round only in the environment it manipulates. And manipulating an environment having more information (i.e. more interactions or operations) can only lead to adding (and not removing) new information into the new environment being computed (i.e. more interactions or operations). 
	
	%Thus, it follows that $\Phi_{i-1}(p)\sqsubseteq\Phi_i(p)$ for all $i$. 
	Now, for a predicate $p/n$, the number of interactions in $\Phi_i(p)$ (for any $i$) is limited by the number of pairs of (possibly interacting) arguments, which is of the order $O(n^2)$. Likewise, the set of operations labeling an interaction is necessarily finite, as its size is limited by the number of program points.
	What remains to be shown is that for an operation (a predicate call, say to some predicate $q/m$) at a given program point, there is no infinite succession of different $\psi$-based operations representing this operation. Now, this could only happen if the called predicate $q/m$ was itself re-analyzed between analysis rounds of $p$. This is excluded, as we restricted programs to direct-recursive programs only, and our analysis analyses predicates bottom-up in the call-graph such that when a predicate is analyzed that is calling $q/m$, the analysis results for $q/m$ are definitely known and hence the $\psi$-based operation representing this call will always be the same (some abstract profile $\psi(\alpha_1,\dots,\alpha_{m})$ or $\psi_{\bot}$ in case the call is recursive). 
\end{proof}

%The interested reader is referred to Appendix B for the proof of Proposition~\ref{prop:convergent}, and to Appendix C for a preliminary result on the analysis soundness. In a likewise manner, Appendix D develops a preliminary result on its worst-case time complexity. 
%%% A general call graph 
\begin{proposition}\label{thm:complexity}
	Let $P$ be a program containing $\ell_P$ predicates, with a total of $\ell_a$ program points. Let $\ell_{io} = \max\{(j+(l-1))\times l \mid p/n \in P, p/n $ has $j$ input arguments and $l$ output arguments$\}$. Then the running time of the analysis is of worst-case complexity $\mathcal{O}(\ell_P\times\ell_{io}\times\ell_a\times\ell_R)$ with $\ell_R$ a finite natural proportional to the number of potential operations to be registered in the predicates. 
\end{proposition}
\begin{proof}
	Let us consider the analysis of a given predicate $p_k/n (k \in 1..\ell_P)$. The required lattice for the abstract value associated to the predicate has $\bot$, i.e. $\{\}$, as minimal set of interactions. The maximal element, $\top_p$, is the set containing an interaction $V_i\xrsquigarrow{\mathbb{O}_k}V_o$ for each pair of variables $V_i, V_o \in \args(p_k)$ such that $i \neq o$ and $V_o$ is output. The elements in-between in the lattice are the sets of "incomplete" interactions, i.e. where all variables and/or operations are not present.
	
	The number of combination of arguments in potential interactions of $p_k$ is $(j+(l-1))\times l$, with $j$, resp. $l$, the number of input, resp. output arguments of $p_k$, since each input argument can have exactly one interaction with each output argument, and each output argument can also contribute to the construction of the $(l-1)$ other output arguments. This quantity is majored by ${n-1} \times n$. 
	
	We still need to prove that a finite number of (also finite) operations from $\lfp(R)$ suffices to populate the potential interactions and thereby restrict the lattice's height. First, observe that the number of operations in an interaction is majored by the number of program point in $P$ which is finite. Now concerning the $\psi$-based operations, only a finite amount of these is treated by the analysis as stated earlier. We will denote by $\ell_R$ the number of operations that the analysis could possibly compute for a predicate given a program's call graph. For $p_k$, this quantity is proportional to both the number of program points in its body and, recursively, the number of potential operations of the predicates it makes calls to. These $\psi$-based operations evolve as they are recomputed by successive analysis rounds; $\ell_R$ represents the number of such steps that can occur before a computed $\psi$-operation converges. %The convergence itself is, of course, guaranteed, since the program call graph cannot contain cycles, so that each predicate's profile is eventually obtained.  
	
	%these exploit previously computed abstract values, which complicates things. However, thanks to our hypothesis forbidding indirect recursion, we know that the predicates that appear as leaf-vertexes in $P$'s call graph only dispose of $\psi_\bot$ which of course is finite. The predicates making calls to such leaves naturally incorporate a finite number of operations in their $\psi$ terms, which ensures that the $\psi$-terms based on these predicates, in turn, are of finite in size; and so on for the entire call graph. We will denote by $\ell_R$ the number of operations that the analysis could possibly compute given a program's call graph. For $p_k$, this quantity is proportional to both the number of program points in its body and, recursively, the number of potential operations of the predicates it makes calls to. 
	
	%Observe that, thanks to our hypothesis forbidding indirect recursion, the analysis can in fact be run on the different predicates in $P$ following $P$'s call graph, treating each time a random leaf-vertex then removing it from the graph. This would not alter the analysis  results.   
	%Now observe that the set $\mathbb{O}_k$ is necessarily finite, being majored by the number of program points in $P$ which is finite.
	So the height of the lattice, that is the maximal number of steps from $\bot$ to $\top_p$, is majored by $\ell_R \times (n - 1) \times  n$ (this corresponds to adding, at each step up the lattice, an operation to one of the existing interactions, or creating an interaction decorated by one operation). As the analysis climbs up in the lattice until reaching a fixpoint, this gives a realistic upper bound for the number of analysis iterations for $p_k$. 
	The analysis might have to run up the lattice of each of the $\ell_P$ predicates in $P$, and at each iteration it needs to crawl through $\ell_a$ program points and compute $\ell_P$ projections, hence the result. 
\end{proof}

%%%%	\begin{figure}[]
%		\centering
%		\begin{tikzpicture} [
	%			->, % makes the edges directed
	%			%	>=stealth’, % makes the arrow heads bold
	%			node distance=1cm, % specifies the minimum distance between two nodes. Change if necessary.
	%			every state/.style={thick, fill=gray!10}, % sets the properties for each ’state’ node
	%			initial text=$ $, % sets the text that appears on the start arrow
	%			]
	%			\node[state] (1) {$p^k_2$};
	%			\node[state, right of=1, xshift=15.0] (3) {$p^{k-1}_2$};
	%			\node[state, above of=3] (9) {$p^{k-1}_1$};
	%			\node[state, below of=3] (10) {$p^{k-1}_{\dots}$};
	%			\node[state, above of=1] (7) {$p^k_1$};
	%			\node[state, below of=1] (8) {$p^k_{\dots}$};
	%			\node[right of=3, dotted] (11) {$\dots$};
	%			
	%			\draw (1) edge[loop left] node{} (1);
	%			\draw (7) edge[loop left] node{} (7);
	%			\draw (8) edge[loop left] node{} (8);
	%			\draw (1) edge[above] node{} (3);
	%			\draw (1) edge[above] node{} (9);
	%			\draw (1) edge[above] node{} (10);
	%			\draw (7) edge[above] node{} (3);
	%			\draw (7) edge[above] node{} (9);
	%			\draw (7) edge[above] node{} (10);
	%			\draw (8) edge[above] node{} (3);
	%			\draw (8) edge[above] node{} (9);
	%			\draw (8) edge[above] node{} (10);
	%		\end{tikzpicture}
%		\caption{A general call graph} 
%		\label{fig:callgraph}
%	\end{figure}

%%%%%%%%%%%%%%%%%%%%%%%%%%%%%%%%%%%%%%%%%%%%%%%%%%%
%%%%%%%%%%%%%%%%%%%%%%%%%%%%%%%%%%%%%%%% 5
%%%%%%%%%%%%%%%%%%%%%%%%%%%%%%%%%%%%%%%%%%%%%%%%%%%
%\input{5-characteristics.tex} 

%%%%%%%%%%%%%%%%%%%%%%%%%%%%%%%%%%%%%%%%%%%%%%%%%%%
%%%%%%%%%%%%%%%%%%%%%%%%%%%%%%%%%%%%%%%% 6
%%%%%%%%%%%%%%%%%%%%%%%%%%%%%%%%%%%%%%%%%%%%%%%%%%%
\section{Conclusions and Future Work}\label{sec:ccl}
This work aims to develop a tractable process for profiling predicate arguments and normalizing their order of apparition in a prototypical Mercury-like language. Our analysis essentially computes a high-level abstraction of program derivations, called \textit{interactions}. Although a normalization procedure already existed for Mercury~\cite{degrave2008}, it focused on normalizing clause bodies and did not address predicate arguments.

Our approach to code normalization revolves around the search of an ordering among predicate arguments. Central to this technique is the research for an \textit{ideal} ordering of the arguments, i.e. a total order $\le$ that allows to sort arguments in a non-ambiguous, unique way, at least in the context of a single program. While we have introduced a first working, but rather arbitrary, example of such an order based on argument profiles metrics, it is our belief that more precise or application-tailored orderings could be found to enhance the analysis output in concrete situations. In particular, identifying the situations in which an order is to be preferred over other incarnations, is left as future work. 

Having a normal form for programs is recognized as an important step in several applications, one of interest being a clone detection scheme, where recognizing a couple of similar predicates implies finding a mapping of clauses and a mapping of arguments among the predicates such that two clauses, or arguments, in the mapping play similar roles in the predicate's definition. The problem, which is intractable in general, becomes radically more manageable if a quadratic approximation is found for one of the two interleaved matching problems~\cite{iwsc}. We intend to explore the use of our analysis for computing a matching of arguments in this context.

Program comprehension is a rising research field in which all aspects of dataflow information constitute useful pieces of information. Program slicing, for example, is a way of extracting the computations in which a given (set of) argument(s) plays a prominent role~\cite{slicing}. Interestingly, what we achieve by computing argument profiles resembles the extraction of such program slices. In existing program slicing techniques however, the computed slices are actual parts of the considered program~\cite{slicing}, whereas our profiles rather constitute abstract representations of data flow information. Moreover, while an argument profile typically exhibits the details of the operations (be it unifications or calls to predicates) that involve the argument, the program portions obtained by means of slicing do not carry any \textit{interpretation} of the program, as the slices' purpose is to represent the part of the program that might be of interest~\cite{slicing-clp}. As an example, consider a predicate in which all of the arguments are somehow participating in every single atom but in different manners. The slices for the different arguments then systematically come down to the whole predicate definition. In contrast, our argument profiles contain finer-grained distinctions, allowing to identify which operations involve which arguments, as well as specific links between input and output arguments -- but abstracting from the order in which the involved atoms are executed. We therefore believe our approach to be complementary to program slicing and to constitute a new step towards better understanding links between arguments and, hence, deriving useful information about the operations hidden in a predicate definition. 

Other analyses addressing program comprehension or security concerns by studying interactions among variables could benefit from our method, some examples being feature analysis, trace analysis and taint analysis~\cite{feature-analysis,pc-survey}.%~\cite{hybrid-taint-analysis} %the latter's purpose being to find out which variables, if marked as tainted before the program executes, can affect (or "taint") other variables during the execution.
%Pushing this study of argument profiles further is an important step to be taken in order to better measure what outcomes the analysis can have in this area of research. \wim{Pas fan de cette dernière phrase, peut-être reformuler "explorer les possibilités de l'analyse dans le contexte de program coprehension" mais cela me semble suggérer qu'on a une analyse avec laquelle on ne sait pas trop quoi faire.} 

%\wim{Paragraphe sur taint analysis est un peu faible. Soit on cite taint analysis comme l'une des applications de notre anlayse, soit on développe l'idée plus loin. Je ne suis pas sûr que ce papier est l'endroit pour développer l'idée} % todo le citer avec les refs comme une applications, as un § à part entière
%%The argument interactions computed by our analysis can also serve in taint analysis, where the purpose is to find out which variables, if marked as tainted before the program executes, can affect (or "taint") other variables during the execution. A program annotated by preliminary information about taint- or untaintedness of input arguments can therefore be analyzed to determine which security breaches, if any, it opens to the outside world. Again, while it is often program slicing that is used in conjunction with taint analyses~\cite{hybrid-taint-analysis,tzslicer}, we believe that our approach can play a useful, complementary role. Building and extending a prototype implementing a taint analysis instance of our profiling model is left for future research. 

%\section*{Acknowledgements}

\nocite{*}
\bibliographystyle{eptcs}
\bibliography{main}

\begin{thebibliography}{10}
\providecommand{\bibitemdeclare}[2]{}
\providecommand{\surnamestart}{}
\providecommand{\surnameend}{}
\providecommand{\urlprefix}{Available at }
\providecommand{\url}[1]{\texttt{#1}}
\providecommand{\href}[2]{\texttt{#2}}
\providecommand{\urlalt}[2]{\href{#1}{#2}}
\providecommand{\doi}[1]{doi:\urlalt{https://doi.org/#1}{#1}}
\providecommand{\eprint}[1]{arXiv:\urlalt{https://arxiv.org/abs/#1}{#1}}
\providecommand{\bibinfo}[2]{#2}

\bibitemdeclare{inproceedings}{tracing-explaining-clpfd}
\bibitem{tracing-explaining-clpfd}
\bibinfo{author}{Magnus \surnamestart {\AA}gren\surnameend},
  \bibinfo{author}{Tam{\'{a}}s \surnamestart Szeredi\surnameend},
  \bibinfo{author}{Nicolas \surnamestart Beldiceanu\surnameend} \&
  \bibinfo{author}{Mats \surnamestart Carlsson\surnameend}
  (\bibinfo{year}{2002}): \emph{\bibinfo{title}{Tracing and Explaining
  Execution of {CLP(FD)} Programs}}.
\newblock In \bibinfo{editor}{Alexandre \surnamestart Tessier\surnameend},
  editor: {\slshape \bibinfo{booktitle}{Proceedings of the 12th International
  Workshop on Logic Programming Environments}}, pp. \bibinfo{pages}{1--16},
  \doi{10.48550/arXiv.cs/0207047}.

\bibitemdeclare{article}{normalization}
\bibitem{normalization}
\bibinfo{author}{Danilo \surnamestart Bruschi\surnameend},
  \bibinfo{author}{Lorenzo \surnamestart Martignoni\surnameend} \&
  \bibinfo{author}{Mattia \surnamestart Monga\surnameend}
  (\bibinfo{year}{2007}): \emph{\bibinfo{title}{Code Normalization for
  Self-Mutating Malware}}.
\newblock {\slshape \bibinfo{journal}{IEEE Security \& Privacy}}
  \bibinfo{volume}{5}(\bibinfo{number}{2}), pp. \bibinfo{pages}{46--54},
  \doi{10.1109/MSP.2007.31}.

\bibitemdeclare{inproceedings}{iterative-dataflow}
\bibitem{iterative-dataflow}
\bibinfo{author}{Keith~D. \surnamestart Cooper\surnameend},
  \bibinfo{author}{Timothy~J. \surnamestart Harvey\surnameend} \&
  \bibinfo{author}{Ken \surnamestart Kennedy\surnameend}
  (\bibinfo{year}{2006}): \emph{\bibinfo{title}{An Empirical Study of Iterative
  Data-Flow Analysis}}.
\newblock In: {\slshape \bibinfo{booktitle}{2006 15th International Conference
  on Computing}}, pp. \bibinfo{pages}{266--276}, \doi{10.1109/CIC.2006.22}.

\bibitemdeclare{article}{pc-survey}
\bibitem{pc-survey}
\bibinfo{author}{Bas \surnamestart Cornelissen\surnameend},
  \bibinfo{author}{Andy \surnamestart Zaidman\surnameend},
  \bibinfo{author}{Arie \surnamestart Deursen\surnameend},
  \bibinfo{author}{Leon \surnamestart Moonen\surnameend} \&
  \bibinfo{author}{Rainer \surnamestart Koschke\surnameend}
  (\bibinfo{year}{2009}): \emph{\bibinfo{title}{A Systematic Survey of Program
  Comprehension through Dynamic Analysis}}.
\newblock {\slshape \bibinfo{journal}{Software Engineering, IEEE Transactions
  on}} \bibinfo{volume}{35}, pp. \bibinfo{pages}{684 -- 702},
  \doi{10.1109/TSE.2009.28}.

\bibitemdeclare{article}{normal-form-asp}
\bibitem{normal-form-asp}
\bibinfo{author}{Stefania \surnamestart Costantini\surnameend} \&
  \bibinfo{author}{Alessandro \surnamestart Provetti\surnameend}
  (\bibinfo{year}{2005}): \emph{\bibinfo{title}{Normal forms for Answer Sets
  Programming}}.
\newblock {\slshape \bibinfo{journal}{Theory and Practice of Logic
  Programming}} \bibinfo{volume}{5}, \doi{10.1017/S1471068404002339}.

\bibitemdeclare{inproceedings}{DandoisV12}
\bibitem{DandoisV12}
\bibinfo{author}{Céline \surnamestart Dandois\surnameend} \&
  \bibinfo{author}{Wim \surnamestart Vanhoof\surnameend}
  (\bibinfo{year}{2012}): \emph{\bibinfo{title}{{Semantic Code Clones in Logic
  Programs}}}.
\newblock In \bibinfo{editor}{E.~\surnamestart Albert\surnameend}, editor:
  {\slshape \bibinfo{booktitle}{Proc. of the 22nd International Symposium on
  Logic-Based Program Synthesis and Transformation ({LOPSTR}'12)}}, {\slshape
  \bibinfo{series}{LNCS}} \bibinfo{volume}{7844},
  \bibinfo{publisher}{Springer}, pp. \bibinfo{pages}{35--50},
  \doi{10.1007/978-3-642-38197-3\_4}.

\bibitemdeclare{article}{efficient-dataflow}
\bibitem{efficient-dataflow}
\bibinfo{author}{Saumya~K. \surnamestart Debray\surnameend}
  (\bibinfo{year}{1992}): \emph{\bibinfo{title}{Efficient Dataflow Analysis of
  Logic Programs}}.
\newblock {\slshape \bibinfo{journal}{J. ACM}}
  \bibinfo{volume}{39}(\bibinfo{number}{4}), p. \bibinfo{pages}{949–984},
  \doi{10.1145/146585.146624}.

\bibitemdeclare{inproceedings}{degrave2008}
\bibitem{degrave2008}
\bibinfo{author}{Fran{\c{c}}ois \surnamestart Degrave\surnameend} \&
  \bibinfo{author}{Wim \surnamestart Vanhoof\surnameend}
  (\bibinfo{year}{2008}): \emph{\bibinfo{title}{Towards a Normal Form for
  Mercury Programs}}.
\newblock In \bibinfo{editor}{Andy \surnamestart King\surnameend}, editor:
  {\slshape \bibinfo{booktitle}{Logic-Based Program Synthesis and
  Transformation}}, \bibinfo{publisher}{Springer}, pp. \bibinfo{pages}{43--58},
  \doi{10.1007/978-3-540-78769-3\_4}.

\bibitemdeclare{article}{feature-analysis}
\bibitem{feature-analysis}
\bibinfo{author}{Thomas \surnamestart Eisenbarth\surnameend},
  \bibinfo{author}{Rainer \surnamestart Koschke\surnameend} \&
  \bibinfo{author}{Daniel \surnamestart Simon\surnameend}
  (\bibinfo{year}{2001}): \emph{\bibinfo{title}{Aiding Program Comprehension by
  Static and Dynamic Feature Analysis}}.
\newblock \doi{10.1109/ICSM.2001.972777}.

\bibitemdeclare{article}{anteater}
\bibitem{anteater}
\bibinfo{author}{Rebecca \surnamestart Faust\surnameend},
  \bibinfo{author}{Katherine \surnamestart Isaacs\surnameend},
  \bibinfo{author}{William~Z. \surnamestart Bernstein\surnameend},
  \bibinfo{author}{Michael \surnamestart Sharp\surnameend} \&
  \bibinfo{author}{Carlos \surnamestart Scheidegger\surnameend}
  (\bibinfo{year}{2019}): \emph{\bibinfo{title}{Anteater: Interactive
  Visualization for Program Understanding}}.
\newblock {\slshape \bibinfo{journal}{CoRR}}, \doi{10.48550/arXiv.1907.02872}.

\bibitemdeclare{article}{lp-semantics}
\bibitem{lp-semantics}
\bibinfo{author}{Melvin \surnamestart Fitting\surnameend}
  (\bibinfo{year}{2002}): \emph{\bibinfo{title}{{Fixpoint semantics for logic
  programming a survey}}}.
\newblock {\slshape \bibinfo{journal}{Theoretical Computer Science}}
  \bibinfo{volume}{278}(\bibinfo{number}{1}), pp. \bibinfo{pages}{25 -- 51},
  \doi{10.1016/S0304-3975(00)00330-3}.
\newblock \bibinfo{note}{Mathematical Foundations of Programming Semantics
  1996}.

\bibitemdeclare{article}{mercury}
\bibitem{mercury}
\bibinfo{author}{Fergus \surnamestart Henderson\surnameend},
  \bibinfo{author}{Thomas \surnamestart Conway\surnameend},
  \bibinfo{author}{Zoltan \surnamestart Somogyi\surnameend},
  \bibinfo{author}{Peter \surnamestart Schachte\surnameend},
  \bibinfo{author}{Simon \surnamestart Taylor\surnameend} \&
  \bibinfo{author}{Chris \surnamestart Speirs\surnameend}
  (\bibinfo{year}{1999}): \emph{\bibinfo{title}{The Mercury Language Reference
  Manual}}.

\bibitemdeclare{inproceedings}{inputtracer}
\bibitem{inputtracer}
\bibinfo{author}{Ulf \surnamestart Kargén\surnameend} \&
  \bibinfo{author}{Nahid \surnamestart Shahmehri\surnameend}
  (\bibinfo{year}{2012}): \emph{\bibinfo{title}{InputTracer: A Data-Flow
  Analysis Tool for Manual Program Comprehension of x86 Binaries}}.
\newblock In: {\slshape \bibinfo{booktitle}{2012 IEEE 12th International
  Working Conference on Source Code Analysis and Manipulation}}, pp.
  \bibinfo{pages}{138--143}, \doi{10.1109/SCAM.2012.16}.

\bibitemdeclare{inproceedings}{traceclp}
\bibitem{traceclp}
\bibinfo{author}{Ludovic \surnamestart Langevine\surnameend},
  \bibinfo{author}{Pierre \surnamestart Deransart\surnameend},
  \bibinfo{author}{Mireille \surnamestart Ducass{\'{e}}\surnameend} \&
  \bibinfo{author}{Erwan \surnamestart Jahier\surnameend}
  (\bibinfo{year}{2001}): \emph{\bibinfo{title}{Prototyping {CLP(FD)} tracers:
  a trace model and an experimental validation environment}}.
\newblock In \bibinfo{editor}{Anthony~J. \surnamestart Kusalik\surnameend},
  editor: {\slshape \bibinfo{booktitle}{Proceedings of the Eleventh Workshop on
  Logic Programming Environments}}, \doi{10.48550/arXiv.cs/0111043}.

\bibitemdeclare{inproceedings}{hybrid-taint-analysis}
\bibitem{hybrid-taint-analysis}
\bibinfo{author}{Florian~D. \surnamestart Loch\surnameend},
  \bibinfo{author}{Martin \surnamestart Johns\surnameend},
  \bibinfo{author}{Martin \surnamestart Hecker\surnameend},
  \bibinfo{author}{Martin \surnamestart Mohr\surnameend} \&
  \bibinfo{author}{Gregor \surnamestart Snelting\surnameend}
  (\bibinfo{year}{2020}): \emph{\bibinfo{title}{Hybrid Taint Analysis for Java
  EE}}.
\newblock In: {\slshape \bibinfo{booktitle}{Proceedings of the 35th Annual ACM
  Symposium on Applied Computing}}, \bibinfo{series}{SAC '20},
  \bibinfo{publisher}{Association for Computing Machinery},
  \bibinfo{address}{New York, NY, USA}, p. \bibinfo{pages}{1716–1725},
  \doi{10.1145/3341105.3373887}.

\bibitemdeclare{article}{live-structure-dataflow}
\bibitem{live-structure-dataflow}
\bibinfo{author}{Anne \surnamestart Mulkers\surnameend},
  \bibinfo{author}{William \surnamestart Winsborough\surnameend} \&
  \bibinfo{author}{Maurice \surnamestart Bruynooghe\surnameend}
  (\bibinfo{year}{1994}): \emph{\bibinfo{title}{Live-Structure Dataflow
  Analysis for Prolog}}.
\newblock {\slshape \bibinfo{journal}{ACM Trans. Program. Lang. Syst.}}
  \bibinfo{volume}{16}(\bibinfo{number}{2}), p. \bibinfo{pages}{205–258},
  \doi{10.1145/174662.174664}.

\bibitemdeclare{article}{parallel-lp}
\bibitem{parallel-lp}
\bibinfo{author}{Kalyan \surnamestart Muthukumar\surnameend},
  \bibinfo{author}{Francisco \surnamestart Bueno\surnameend},
  \bibinfo{author}{Maria~Jose \surnamestart {García de la Banda}\surnameend}
  \& \bibinfo{author}{Manuel \surnamestart Hermenegildo\surnameend}
  (\bibinfo{year}{1999}): \emph{\bibinfo{title}{Automatic compile-time
  parallelization of logic programs for restricted, goal level, independent and
  parallelism}}.
\newblock {\slshape \bibinfo{journal}{The Journal of Logic Programming}}
  \bibinfo{volume}{38}(\bibinfo{number}{2}), pp. \bibinfo{pages}{165--218},
  \doi{10.1016/S0743-1066(98)10022-5}.

\bibitemdeclare{inproceedings}{blanker}
\bibitem{blanker}
\bibinfo{author}{Davide \surnamestart Pizzolotto\surnameend} \&
  \bibinfo{author}{Katsuro \surnamestart Inoue\surnameend}
  (\bibinfo{year}{2020}): \emph{\bibinfo{title}{Blanker: A Refactor-Oriented
  Cloned Source Code Normalizer}}.
\newblock In: {\slshape \bibinfo{booktitle}{14th International Workshop on
  Software Clones}}, pp. \bibinfo{pages}{22--25},
  \doi{10.1109/IWSC50091.2020.9047634}.

\bibitemdeclare{article}{rattansurvey}
\bibitem{rattansurvey}
\bibinfo{author}{Dhavleesh \surnamestart Rattan\surnameend},
  \bibinfo{author}{Rajesh \surnamestart Bhatia\surnameend} \&
  \bibinfo{author}{Maninder \surnamestart Singh\surnameend}
  (\bibinfo{year}{2013}): \emph{\bibinfo{title}{Software clone detection: A
  systematic review}}.
\newblock {\slshape \bibinfo{journal}{Information and Software Technology}}
  \bibinfo{volume}{55}(\bibinfo{number}{7}), pp. \bibinfo{pages}{1165--1199},
  \doi{10.1016/j.infsof.2013.01.008}.

\bibitemdeclare{article}{slicing-clp}
\bibitem{slicing-clp}
\bibinfo{author}{Gyöngyi \surnamestart Szilágyi\surnameend},
  \bibinfo{author}{Tibor \surnamestart Gyimóthy\surnameend} \&
  \bibinfo{author}{Jan \surnamestart Maluszynski\surnameend}
  (\bibinfo{year}{2002}): \emph{\bibinfo{title}{Static and Dynamic Slicing of
  Constraint Logic Programs}}.
\newblock {\slshape \bibinfo{journal}{Automated Software Engineering}}
  \bibinfo{volume}{9}, pp. \bibinfo{pages}{41--65},
  \doi{10.1023/A:1013280119003}.

\bibitemdeclare{article}{compiling-dataflow}
\bibitem{compiling-dataflow}
\bibinfo{author}{Jichang \surnamestart Tan\surnameend} \&
  \bibinfo{author}{I-Peng \surnamestart Lin\surnameend} (\bibinfo{year}{1992}):
  \emph{\bibinfo{title}{Compiling Dataflow Analysis of Logic Programs}}.
\newblock {\slshape \bibinfo{journal}{SIGPLAN Not.}}
  \bibinfo{volume}{27}(\bibinfo{number}{7}), p. \bibinfo{pages}{106–115},
  \doi{10.1145/143103.143123}.

\bibitemdeclare{inproceedings}{wim-bta}
\bibitem{wim-bta}
\bibinfo{author}{Wim \surnamestart Vanhoof\surnameend} (\bibinfo{year}{2000}):
  \emph{\bibinfo{title}{Binding-Time Analysis by Constraint Solving}}.
\newblock In \bibinfo{editor}{Michel \surnamestart Parigot\surnameend} \&
  \bibinfo{editor}{Andrei \surnamestart Voronkov\surnameend}, editors:
  {\slshape \bibinfo{booktitle}{Logic for Programming and Automated
  Reasoning}}, \bibinfo{publisher}{Springer}, \bibinfo{address}{Berlin,
  Heidelberg}, pp. \bibinfo{pages}{399--416}, \doi{10.1007/3-540-44404-1\_25}.

\bibitemdeclare{inproceedings}{clones}
\bibitem{clones}
\bibinfo{author}{Wim \surnamestart Vanhoof\surnameend} \&
  \bibinfo{author}{Gonzague \surnamestart Yernaux\surnameend}
  (\bibinfo{year}{2020}): \emph{\bibinfo{title}{Generalization-Driven Semantic
  Clone Detection in CLP}}.
\newblock In \bibinfo{editor}{Maurizio \surnamestart Gabbrielli\surnameend},
  editor: {\slshape \bibinfo{booktitle}{Logic-Based Program Synthesis and
  Transformation}}, \bibinfo{publisher}{Springer International Publishing},
  \bibinfo{address}{Cham}, pp. \bibinfo{pages}{228--242},
  \doi{10.1007/978-3-030-45260-5\_14}.

\bibitemdeclare{article}{slicing}
\bibitem{slicing}
\bibinfo{author}{Martin \surnamestart Ward\surnameend} \&
  \bibinfo{author}{Hussein \surnamestart Zedan\surnameend}
  (\bibinfo{year}{2007}): \emph{\bibinfo{title}{Slicing as a program
  transformation}}.
\newblock {\slshape \bibinfo{journal}{ACM Trans. Program. Lang. Syst.}}
  \bibinfo{volume}{29}, \doi{10.1145/1216374.1216375}.

\bibitemdeclare{phdthesis}{parallel-lp-thesis}
\bibitem{parallel-lp-thesis}
\bibinfo{author}{William~Hale \surnamestart Winsborough\surnameend} \&
  \bibinfo{author}{Charles~N. \surnamestart Fischer\surnameend}
  (\bibinfo{year}{1988}): \emph{\bibinfo{title}{Automatic, Transparent
  Parallelization of Logic Programs at Compile Time}}.
\newblock Ph.D. thesis, \bibinfo{school}{The University of Wisconsin -
  Madison}.

\bibitemdeclare{inproceedings}{tzslicer}
\bibitem{tzslicer}
\bibinfo{author}{Mengmei \surnamestart Ye\surnameend},
  \bibinfo{author}{Jonathan \surnamestart Sherman\surnameend},
  \bibinfo{author}{Witawas \surnamestart Srisa-an\surnameend} \&
  \bibinfo{author}{Sheng \surnamestart Wei\surnameend} (\bibinfo{year}{2018}):
  \emph{\bibinfo{title}{TZSlicer: Security-aware dynamic program slicing for
  hardware isolation}}.
\newblock In: {\slshape \bibinfo{booktitle}{2018 IEEE International Symposium
  on Hardware Oriented Security and Trust (HOST)}}, pp.
  \bibinfo{pages}{17--24}, \doi{10.1109/HST.2018.8383886}.

\bibitemdeclare{article}{gen}
\bibitem{gen}
\bibinfo{author}{Gonzague \surnamestart Yernaux\surnameend} \&
  \bibinfo{author}{Wim \surnamestart Vanhoof\surnameend}
  (\bibinfo{year}{2019}): \emph{\bibinfo{title}{{Anti-unification in Constraint
  Logic Programming}}}.
\newblock {\slshape \bibinfo{journal}{Theory and Practice of Logic
  Programming}} \bibinfo{volume}{19}(\bibinfo{number}{5-6}), p.
  \bibinfo{pages}{773–789}, \doi{10.1017/S1471068419000188}.

\bibitemdeclare{inproceedings}{iwsc}
\bibitem{iwsc}
\bibinfo{author}{Gonzague \surnamestart Yernaux\surnameend} \&
  \bibinfo{author}{Wim \surnamestart Vanhoof\surnameend}
  (\bibinfo{year}{2022}): \emph{\bibinfo{title}{On Detecting Semantic Clones in
  Constraint Logic Programs}}.
\newblock In: {\slshape \bibinfo{booktitle}{2022 IEEE 16th International
  Workshop on Software Clones (IWSC)}}, pp. \bibinfo{pages}{32--38},
  \doi{10.1109/IWSC55060.2022.00014}.

\end{thebibliography}

%\newpage 
%\section*{Appendices}
%\input{appendices.tex}

\end{document}